\documentclass[twoside,bezier,12pt, reqno]{amsart}
\usepackage{amssymb,latexsym,epsfig,mathrsfs,graphpap,graphics,amsmath,amssymb}
\usepackage{amstext}
\usepackage{amsthm}
\usepackage[utf8]{inputenc}
\usepackage[english]{babel}
\usepackage{helvet}
\usepackage{courier}

\newtheorem{theorem}{Theorem}[section]
\newtheorem{lemma}{Lemma}[section]

\theoremstyle{Definition}
\newtheorem{definition}{Definition}[section]

\theoremstyle{remark}
\newtheorem{remark}[theorem]{Remark}
\numberwithin{equation}{section}

\begin{document}

\begin{flushleft}
 {\bf\Large { Uncertainty Principles for the  Short-time Non-separable Linear Canonical Transform }}

\parindent=0mm \vspace{.2in}

{\bf{M. Younus Bhat$^{1},$ and Aamir H. Dar$^{2}$ }}
\end{flushleft}

{{\it $^{1}$ Department of  Mathematical Sciences,  Islamic University of Science and Technology Awantipora, Pulwama, Jammu and Kashmir 192122, India.E-mail: $\text{g gyounusg@gmail.com}$}}

{{\it $^{2}$ Department of  Mathematical Sciences,  Islamic University of Science and Technology Awantipora, Pulwama, Jammu and Kashmir 192122, India.E-mail: $\text{ahdkul740@gmail.com}$}}

\begin{quotation}
\noindent
{\footnotesize {\sc Abstract.} The free metaplectic transformation (FMT) or the nonseparable linear canonical transformation (NSLCT) has gained much popularity in recent times because of its various application in signal processing, paraxial optical systems, digital algorithms, optical encryption and so on. However, the NSLCT is
inadequate for localized analysis of non-transient
signals, as such, it is
imperative to introduce a unique localized transform coined as the short-time  nonseparable linear canonical transformation (ST-NSLCT). In
this paper, we investigate the ST-NSLCT. Firstly, we propose the definition of
the ST-NSLCT, and provide the time-frequency analysis of the proposed transform in
the NSLCT domain. Secondly, we investigate the basic properties of the proposed
transform including the reconstruction formula,
Moyal's formula. The emergence of the ST-NSLCT definition and its properties
broadens the development of time-frequency representation of higher-dimensional
signals theory to a certain extent. Finally, we extend some different uncertainty principles (UP) from quantum mechanics
including Lieb's inequality, Pitt's inequality, Hausdorff-Young inequality,  Heisenberg’s uncertainty principle,  Hardy’s
UP, Beurling’s UP, Logarithmic UP, and Nazarov’s UP.\\

}
\end{quotation}

\section{Introduction} \label{sec intro}
\noindent

The free metaplectic transformation (FMT) also known as the nonseparable linear canonical transformation (NSLCT) is an n-dimensional linear canonical transformation (LCT) first studied in \cite{ND1}, is widely
used in many fields such as filter design, pattern recognition, optics and analyzing the propagation
of electromagnetic waves \cite{ND2,ND3,ND4,ND5}. The theory of NSLCT involving a general $2n\times 2n$ real, symplectic matrix ${\bf M}=(A,B:C,D)$ with $n(2n+1)$ degrees
of freedom \cite{WD3,WD4}. The NSLCT embodies several signal processing tools ranging from the classical
Fourier, Fresnel transform, and even the fundamental operations of quadratic phase factor multiplication \cite{WD1,WD2,OWNb}. The metaplectic operator or the NSLCT of any function $f\in L^2(\mathbb R^n)$ with the real free symplectic
matrix ${\bf M}=(A,B:C,D)$ is given by \cite{WD2,NDself}

\begin{equation}\label{eqn mul lct A}
 \mathcal L_{\bf M}[f](w)=\int_{\mathbb R^n}f(x)\mathcal K_M(x,w)dx
\end{equation}
where $\mathcal K_{\bf M}(x,w)$ denotes the kernel and is given by

\begin{equation}\label{kera}
\mathcal K_{\bf M}(x,w)=\frac{1}{(2\pi)^{n/2}\sqrt{|det(B)|}}e^{\frac{i}{2}(w^TDB^{-1}w-2w^TB^{-T}x+x^TB^{-1}Ax)},\quad x,w\in\mathbb R^n, |det (B)|\ne 0.
\end{equation}
The  arbitrary real
parameters involved in (\ref{eqn mul lct A}) are of great importance   for the efficient analysis of the inescapable signals i.e. the chirp-like  signals. Due to the
extra degrees of freedom, the NSLCT  has been successfully employed in diverse
problems arising in various branches of science and engineering, such as harmonic analysis,
 optical systems, reproducing kernel Hilbert spaces, quantum mechanics,image
processing, sampling and so on \cite{NSWL6,OWNc,OWNd}. However, NSLCT
has a drawback. Due to its global kernel it is not suitable for processing the signals
with varying frequency content \cite{OWN1,OWN2}. The short-time non-separable linear canonical transform (ST-NSLCT)
with a local window function can efficiently localize the frequency spectrum of non-transient signals in the
non-separable LCT domain, hence overcomes this drawback. Taking this opportunity, our goal is to introduce the notion of short-time
non-separable linear canonical transform, which
is a generalized version of NSLCT and endowed with higher degrees of freedom resulting in an efficient localized analysis of chirp signals.

Let us now move to the side of uncertainty inequality. Uncertainty principle was introduced by German physicists Heisenberg \cite{10own} in 1927 which is known as the heart of any
ignal processing tool. With the passage of time researchers further extended the uncertainty principle to different types of new uncertainty principles associated with the Fourier
transform, for instance Heisenberg’s uncertainty principle, Logarithmic uncertainty principle, Hardy’s uncertainty principle and Beurling’s uncertainty principle[see\cite{ND7}-\cite{ND16}]. Later these
uncertainty principles were extended to LCT and its generalized  domains \cite{ND17}-\cite{OWN4}. In \cite{NDself} authors proposed
uncertainty principles associated with the NSLCT and in \cite{NSLCWT} authors establish
uncertainty principles for the non-separable linear canonical wavelet transform. Keeping in view the fact that the theory of ST-NSLCT and associated uncertainty principles is yet to be investigated exclusively;
therefore, it is both theoretically interesting and practically useful to study the properties of ST-NSLCT and  formulate some new uncertainty inequalities
pertaining to it. \\

The highlights of the article are pointed out below:\\
\begin{itemize}

\item To introduce a novel integral transform coined as the short-time non-separable linear canonical transform.\\

\item To study the fundamental properties of the proposed transform, including the Moyal's formula, boundedness and inversion formula.\\

\item To establish the Pitt’s inequality, Lieb inequality   and Hausdorff-Young inequality
associated with the ST-NSLCT.\\

\item To formulate the Heisenberg's, Logarithmic and Nazarov’s uncertainty principles.\\

\item To formulate the Hardy's and Beurling's uncertainty principles for the ST-NSLCT.\\

\end{itemize}

The paper is organized as follows: In Section \ref{sec 2}, we discuss some preliminary results and definitions which will be used in subsequent sections. In section \ref{sec 3}, we formally introduced  the definition of short-time non-separable linear canonical transform (ST-NSLCT). Then we investigated  several basic properties of the ST-NSLCT  which are important for signal representation in signal processing.  In Section \ref{sec 4}, we extend some different uncertainty principles (UP's) from quantum mechanics
including Lieb's, Pitt's UP, Heisenberg’s uncertainty principle, Hausdorff-Young,  Hardy’s
UP, Beurling’s UP, Logarithmic UP, and Nazarov’s UP which have already been well studied in the
NSLCT(FMT) domain.

\section{Preliminary}\label{sec 2}

This section give some useful definitions and lemmas about the multi-dimensional Fourier transform and  non-separable linear canonical transform.
\subsection{\bf Multi-dimensional Fourier transform (FT)}
For any $f\in L^2(\mathbb R^n)$ the n-dimensional
Fourier transform (FT) of $f(x)$ is given by \cite{NDself}
\begin{equation}\label{n-ft 1}
\mathcal F[f](w)=\frac{1}{(2\pi)^{n/2}}\int_{\mathbb R^n}f(x)e^{-iw^Tx}dx,
\end{equation}

and its inversion is given by
\begin{equation}\label{inv-ft 1}
 f(x)=\frac{1}{(2\pi)^{n/2}}\int_{\mathbb R^n}\mathcal F[f](w)e^{iw^Tx}dw
\end{equation}
where $x=(x_1,x_2,...,x_n)^T\in\mathbb R^n$, $w=(w_1,w_2,...,w_n)^T\in\mathbb R^n.$\\
Based on the definition of short-time Fourier transform (STFT)\cite{bah}, we can define n-dimensional STFT as:

\begin{definition}\label{def STFT} Let $\phi$ be a window function in $L^2(\mathbb R^n)$, then
for any $f\in L^2(\mathbb R^n)$ the n-dimensional short-time
Fourier transform (STFT) of $f(x)$ with respect to the
window function $\phi$ is given by \cite{NSLCWT}
\begin{equation}\label{n-stft 1}
\mathcal V_{\phi}[f](w,u)=\frac{1}{(2\pi)^{n/2}}\int_{\mathbb R^n}f(x)\overline{\phi(x-u)}e^{-iw^Tx}dx.
\end{equation}
\end{definition}

\subsection{\bf The non-separable linear canonical transform }
For typographical convenience, we shall denote a $2n\times 2n$ matrix ${\bf M}= \left(
\begin{array}{cc}
A & B \\
C & D \\
\end{array}
\right)$ as  ${\bf M}=(A,B:C,D),$ where $A$, $B$, $C$ and $D$ are real $n\times n$ sub-matrices. Moreover, we recall that the matrix $\bf M$  is said to be a free
symplectic matrix if ${\bf M}^T\bf \Omega M=\bf \Omega,$ and $|det(B)|\ne 0$, where ${\bf \Omega}=({0,I_n:-I_n,0})$ and $ I_n$ represents the $n\times n$ identity matrix. The sub-matrices of ${\bf M}$ satisfying the following constraints:
\begin{equation*}
AB^T=BA^T,\quad CD^T=DC^T,\quad AD^T-BC^T=I_n.
\end{equation*}
The transpose of ${\bf M}=\left(
\begin{array}{cc}
A & B \\
C & D \\
\end{array}
\right)$ is given by ${\bf M}^T=\left(
\begin{array}{cc}
A^T & C^T \\
B^T & D^T \\
\end{array}
\right).$ \\
Also inverse of the free
symplectic matrix is given by ${\bf M}^{-1}=\left(
\begin{array}{cc}
D^T & -B^T \\
-C^T & A^T \\
\end{array}
\right).$
It is clear that $\bf M M^{-1}=\left(
\begin{array}{cc}
I_n & 0 \\
0 & I_n \\
\end{array}
\right)$.

\begin{definition}[NSLCT]\cite{ND12}\label{def mul lct}The non-separable
linear canonical transform or the  free metaplectic transform of a function $f\in L^2(\mathbb R^n)$ with respect $2n\times 2n $ real, free  symplectic matrix ${\bf M}=(A,B: C,D)$ (with $|det( B)| \ne 0$) is defined by 
\end{definition}
\begin{equation}\label{eqn mul lct}
 \mathcal L_{\bf M}[f](w)=\int_{\mathbb R^n}f(x)\mathcal K_M(x,w)dx
\end{equation}
where $\mathcal K_{\bf M}(x,w)$ denotes the kernel and is given by

\begin{equation}\label{ker}
\mathcal K_{\bf M}(x,w)=\frac{1}{(2\pi)^{n/2}\sqrt{|det(B)|}}e^{\frac{i}{2}(w^TDB^{-1}w-2w^TB^{-T}x+x^TB^{-1}Ax)}
\end{equation}
where $x=(x_1,x_2,...,x_n)^T\in\mathbb R^n$, $w=(w_1,w_2,...,w_n)^T\in\mathbb R^n.$\\
\begin{definition}\label{def inv MLCT}
Suppose $f\in L^2(\mathbb R^n)$, then the inversion of the non-separable linear canonical transform of
$f$ is given by
\begin{equation}\label{inv MLCT}
f(x)=\mathcal L_{\bf M^{-1}}\left\{ \mathcal L_{\bf M}[f](w)\right\}(x)=\int_{\mathbb R^n} \mathcal L_{\bf M}[f](w)\mathcal K_{\bf M^{-1}}(w,x)dw,
\end{equation}
where ${\bf M}^{-1}=(D^T,-B^T : -C^T,A^T).$

\end{definition}
 For real, symplectic matrix ${\bf M}=(A,B: C,D)$, the non-separable linear canonical transform
kernel (\ref{ker}) satisfies the following important properties:\\

(i) $\mathcal K_{\bf M^{-1}}(w,x)=\overline{\mathcal K_{\bf M}(x,w)},$\\

(ii) $\int_{\mathbb R^n}\mathcal K_{\bf M}(x,w)\mathcal K_{\bf M^{-1}}(z,x)dx=\delta(z-w).$\\

(iii)$\int_{\mathbb R^n}\mathcal K_{\bf M}(x,w)\mathcal K_{\bf N}(x,z)dx=\mathcal K_{\bf MN}(w,z).$\\\

\begin{lemma}\label{parseval}
Let $f,g\in L^2(\mathbb R^n)$, then the  non-separable linear canonical transform satisfies the following parseval's formula:
\begin{equation}\label{eqn par}
\langle f,g\rangle_2=\langle \mathcal L_{\bf M}[f],\mathcal L_{\bf M}[g]\rangle_2.
\end{equation}
For $f=g$, one have \begin{equation}\label{eqn par2} \|f\|^2_{L^2(\mathbf R^n)}=\|\mathcal L_{\bf M}[f]\|^2_{L^2(\mathbf R^n)}.
\end{equation}
\end{lemma}

The non-separable linear canonical transform (defined in (\ref{eqn mul lct})) of a function $f \in L^2(\mathbb R^n)$ can be computed via associated n-dimensional FT, namely

\begin{equation}\label{nlct-ft1}
\mathcal L_{\bf M}[f](w)=\frac{e^{\frac{i(w^TDB^{-1}w)}{2}}}{\sqrt{det(B)}}\mathcal F\left\{e^{\frac{i(x^TB^{-1}Ax)}{2}}f(x)\right\}(B^{-1}w),
\end{equation}

where $\mathcal F\{f\}(w)$ represents the n-dimensional FT defined in (\ref{n-ft 1})
\begin{equation}\label{n-ft}
\mathcal F\{f\}(w)=\frac{1}{(2\pi)^{n/2}}\int_{\mathbb R^n}f(x)e^{-iw^Tx}dx.
\end{equation}
Eq.(\ref{nlct-ft1}) can be rewritten as

\begin{equation}\label{nlct-ft2}
e^{\frac{-i(w^TDB^{-1}w)}{2}}\mathcal L_{\bf M}[f](w)=\mathcal F\left\{H\right\}(B^{-1}w),
\end{equation}
where $H(x)$
 is given by
 \begin{equation}\label{fun H}
H(x)=\frac{1}{\sqrt{det(B)}}e^{\frac{i(x^TB^{-1}Ax)}{2}}f(x)
\end{equation}

\begin{lemma}\label{lem NLCT-FT mod}
Let $f\in L^2(\mathbb R^n)$ and $H(x)=\frac{1}{\sqrt{det(B)}}e^{\frac{i(x^TB^{-1}Ax)}{2}}f(x)$,  then we have following relationship between n-dimensional FT and NSLCT
\begin{equation}\label{eqn NLCT-FT mod}
|\mathcal F[H](w)|=|\mathcal L_{\bf M}[f](Bw)|.
\end{equation}
\end{lemma}
\begin{proof}The proof follows by
changing $w=Bw$ and taking modulus in (\ref{nlct-ft2}).\\\\
\end{proof}

\begin{remark}: If we change $w=Bw$ and take modulus in (\ref{nlct-ft1}), we have
\begin{equation}\label{eqn NLCT-FT mod 0}
|\mathcal F[H_0](w)|=\sqrt{det(B)}|\mathcal L_{\bf M}[f](Bw)|,
\end{equation}
where $H_0(x)={e^{\frac{i(x^TB^{-1}Ax)}{2}}}f(x)$.
\end{remark}

\section{ The Short-time Non-separable Linear Canonical Transform}\label{sec 3}
In this section, we introduce the novel short-time non-separable linear canonical transform
(ST-NSLCT) and discuss several basic properties
of the ST-NSLCT. These properties play important roles in signal representation of multi-dimensional signals.

\begin{definition}\label{def M-WLCT} Let ${\bf M}=(A,B: C,D)$ be a real, free symplectic matrix. Let $\phi$ be a window function in $L^2(\mathbb R^n)$, the short-time non-separable linear canonical transform
(ST-NSLCT) of the function $f\in L^2(\mathbb R^n)$ with
respect to $\phi$ is defined by
\begin{equation}\label{eqn M-WLCT}
\mathcal V^{\bf M}_\phi[f](w,u)=\int_{\mathbb R^n}f(x)\overline{\phi(x-u)}\mathcal K_{\bf M}(x,w)dx,
\end{equation}
where $x=(x_1,x_2,...,x_n)\in \mathbb R^n$, $u=(u_1,u_2,...,u_n)\in \mathbb R^n$, $w=(w_1,w_2,...,w_n)\in \mathbb R^n$ and $\mathcal K_{\bf M}(x,w)$ is given by (\ref{ker}).
\end{definition}

 In terms of classic convolution, formula (\ref{eqn M-WLCT}) can be written as
\begin{equation*}
\mathcal V^{\bf M}_\phi[f](w,u)=(f(u)\mathcal K_{\bf M}(u,w))\ast\overline{\phi(-u)}.
\end{equation*}

Also by applying the properties of  non-separable linear canonical transform, (\ref{eqn M-WLCT}) can be rewritten in the form of an
inner product as

\begin{equation}\label{NwLC ip}
\mathcal V^{\bf M}_\phi[f](w,u)=\langle f,\Psi^{\bf M}_{x,w,u}\rangle,
\end{equation}
 where $\Psi^{\bf M}_{x,w,u}=\phi(x-u)\mathcal K_{\bf M}(u,w).$\\

The prolificacy of the short-time non-separable linear canonical transform given in Definition \ref{def M-WLCT} can be ascertained
from the following important deductions:\\

\begin{enumerate}
\item[(i)] When ${\bf M}=(I_n\cos\alpha,I_n\sin\alpha:-I_n\sin\alpha,I_n\cos\alpha)$
the ST-NSLCT (\ref{eqn M-WLCT}) yields the
n-dimensional non-separable short-time fractional Fourier transform:
\begin{equation*}
\mathcal F^{\alpha}_{\phi}[f](w,u)=\frac{1}{(2\pi)^{n/2}|\sin\alpha|^{n/2}}\int_{\mathbb R^n}f(x)\overline{\phi(x-u)}e^{\frac{i}{2}(w^Tw+x^Tx)\cot\alpha-iw^Tx\csc\alpha} dx.\\\\
\end{equation*}

\item[(ii)]  For ${\bf M}=({\bf 0},I_n\:-I_n,{\bf 0})$
the ST-NSLCT (\ref{eqn M-WLCT}) boils to
n-dimensional  short-time  Fourier transform:
\begin{equation*}
\mathcal F_{\phi}[f](w,u)=\frac{1}{(2\pi)^{n/2}}\int_{\mathbb R^n}f(x)\overline{\phi(x-u)}e^{-iw^Tx} dx.\\\\
\end{equation*}

\item[(iii)] If the sub-matrices  $A=diag(a_{11},...,a_{nn})$, $B=diag(b_{11},...,b_{nn})$,  $C=diag(c_{11},...,c_{nn})$ and $D=diag(d_{11},...,d_{nn})$ of the real, symplectic matrix ${\bf M}$ are taken then ST-NSLCT (\ref{eqn M-WLCT})yields short-time separable
linear canonical transform:

\begin{equation*}
\mathcal L^{\bf M}_{\phi}[f](w,u)=\frac{1}{(2\pi)^{n/2}|\Pi_{j=1}^n b_{jj}|^{1/2}}\int_{\mathbb R^n}f(x)\overline{\phi(x-u)}e^{i\sum_{j=1}^n \left(\frac{d_{jj}w^2_{j}-2w_jx_j+a_{jj}x^2_{j}}{2b_{jj}}\right)} dx.\\
\end{equation*}

\item[(iv)]  For ${\bf M}=(I_n,B:{\bf 0},I_n)$
   the ST-NSLCT (\ref{eqn M-WLCT}) boils down to
     n-dimensional  short-time  Fresnel transform:
\begin{equation*}
\mathcal F^{\bf M}_{\phi}[f](w,u)=\frac{1}{(2\pi)^{n/2}\sqrt{|det(B)|}}\int_{\mathbb R^n}f(x)\overline{\phi(x-u)}e^{i\frac{(w^TB^{-1}w-2w^{T}B^{-T}x+x^{T}B^{-1}x)}{2}} dx.\\\\
\end{equation*}
\end{enumerate}
Now we shall establish relationsip between NSLCT and ST-NSLCT, which will be helpful in establishing some UP's in section \ref{sec 4}.\\

For fixed $u,$ we have
\begin{equation}\label{eqn WLCT-LCT}
\mathcal V^{\bf M}_\phi[f](w,u)=\mathcal L_{\bf M}\left\{f(x)\overline{\phi(x-u)}\right\}(w)
\end{equation}
Applying inverse  non-separable linear canonical transform
 (\ref{inv MLCT}), we have
\begin{eqnarray}
\label{f-phi} f(x)\overline{\phi(x-u)}&=&\mathcal L_{\bf  M^{-1}}\left\{\mathcal V^{\bf M}_\phi[f](w,u) \right\}\\
\nonumber&=&\int_{\mathbb R^n}\mathcal V^{\bf M}_\phi[f](w,u)\mathcal K_{\bf M^{-1}}(w,x)dw\\
\label{f-phi1}&=&\int_{\mathbb R^n}\mathcal V^{\bf M}_\phi[f](w,u)\overline{\mathcal K_{\bf M}(x,w)}dw.
\end{eqnarray}

Now, we discuss several basic properties of the ST-NSLCT given by (\ref{eqn M-WLCT}). These properties
play important roles in multi-dimensional signal processing.\\

The properties which follows directly from definition of the ST-NSLCT viz: Linearity; Anti-linearity; Translation; Modulation and Scaling are omitted. We shall focus on boundedness, Moyal's formula and inversion.

\begin{lemma}[ Relation with n-dimensional  STFT]\label{lem WLCT-WFT}

\begin{eqnarray*}
\nonumber&& \mathcal V^{\bf M}_\phi[f](w,u)\\\
\nonumber&&=\frac{1}{(2\pi)^{n/2}\sqrt{|det(B)|}}\int_{\mathbb R^n}f(x)\overline{\phi(x-u)}e^{\frac{i(w^TDB^{-1}w+x^TB^{-1}Ax-2w^TB^{-T}x)}{2}}dx\\\
\nonumber&&=\frac{e^{i\frac{w^TDB^{-1}w}{2}}}{(2\pi)^{n/2}\sqrt{|det(B)|}}\int_{\mathbb R^n}e^{\frac{i(x^TB^{-1}Ax)}{2}}f(x)\overline{\phi(x-u)}e^{-i(B^{-1}w)^Tx}dx\\
\nonumber&&={e^{i\frac{w^TDB^{-1}w}{2}}}\mathcal V_\phi[H](B^{-1}w,u),\\
\end{eqnarray*}
\end{lemma}
where $\mathcal V_{\phi}[f]$ represents n-dimensional STFT given in Definition \ref{def STFT}.\\

 And
 \begin{equation}\label{fun H1}
H(x)=\frac{1}{\sqrt{|det(A)|}}e^{\frac{i(x^TB^{-1}Ax)}{2}}f(x).\\
\end{equation}

Next we prove the lemma which is very important in  establishing various uncertainty principles in Section \ref{sec 4}
\begin{lemma}\label{lem NWLCT-FT mod}
Let $f\in L^2(\mathbb R^n)$ and $F(x)=e^{\frac{i(x^TB^{-1}Ax)}{2}}f(x)\overline{\phi(x-u)}$,  then we have following relationship between n-dimensional FT and ST-NSLCT.
\begin{equation}\label{eqn NWLCT-FT mod}
|\mathcal F[F](w)|=\sqrt{|det(B)|}|\mathcal V^{\bf M}_\phi[f](Bw,u)|.
\end{equation}
\end{lemma}
\begin{proof}
From (\ref{eqn M-WLCT}), we have
\begin{eqnarray*}
\nonumber&& \mathcal V^{\bf M}_\phi[f](w,u)\\\
\nonumber&&=\frac{1}{(2\pi)^{n/2}\sqrt{|det(B)|}}\int_{\mathbb R^n}f(x)\overline{\phi(x-u)}e^{\frac{i(w^TDB^{-1}w+x^TB^{-1}Ax-2w^TB^{-T}x)}{2}}dx\\\
\nonumber&&=\frac{e^{i\frac{w^TDB^{-1}w}{2}}}{(2\pi)^{n/2}\sqrt{|det(B)|}}\int_{\mathbb R^n}e^{\frac{i(x^TB^{-1}Ax)}{2}}f(x)\overline{\phi(x-u)}e^{-i(B^{-1}w)^Tx}dx\\\
\label{a}&&=\frac{1}{\sqrt{|det(B)|}}{e^{i\frac{w^TDB^{-1}w}{2}}}\mathcal F[F](B^{-1}w),\\\
\end{eqnarray*}

On changing $w=Bw$ and taking modulus in (\ref{a}), we get
$$\sqrt{|det(B)|}|\mathcal V^{\bf M}_\phi[f](Bw,u)|=|\mathcal F[F](w)|$$.
Which completes the proof.
 \end{proof}

\begin{remark}: If we take $F_0(x)=\frac{e^{\frac{i(x^TB^{-1}Ax)}{2}}}{\sqrt{det(B)}}f(x)\overline{\phi(x-u)}$, then above lemma yields
\begin{equation}\label{eqn NWLCT-FT  mod 0}
|\mathcal F[F_0](w)|=|\mathcal L_{\bf M}[f](Bw)|.\\
\end{equation}
\end{remark}

\begin{theorem}[Boundedness] Let  $f,\phi\in L^2(\mathbb R^n)$, where $\phi$ is a non zero window function then we have
\begin{equation*}
\left|\mathcal V^{\bf M}_\phi[f](w,u)\right|\le\frac{1}{(2\pi)^{n/2}\sqrt{|det(B)|}}\|f\|_{L^2(\mathbb R^n)}\|\phi\|_{L^2(\mathbb R^n)},
\end{equation*}
\end{theorem}
\begin{proof}
By using the  Cauchy-Schwarz inequality in Definition \ref{def M-WLCT}, we have
\begin{eqnarray*}
&&\left|\mathcal V^{\bf M}_\phi[f](w,u)\right|^2\\\
&&=\left|\int_{\mathbb R^n}f(x)\overline{\phi(x-u)}\mathcal K_{\bf M}(x,w)dx\right|^2\\\
&&\le\left(\int_{\mathbb R^n}\left|f(x)\overline{\phi(x-u)}\mathcal K_{\bf M}(x,w)\right|dx\right)^2\\\
&&=\left(\frac{1}{(2\pi)^{n/2}\sqrt{|det(B)|}}\int_{\mathbb R^n}\left|f(x)\overline{\phi(x-u)}e^{\frac{i(w^TDB^{-1}w+x^TB^{-1}Ax-2w^TB^{-T}x)}{2}}\right|dx\right)^2\\\
&&=\left(\frac{1}{(2\pi)^{n/2}\sqrt{|det(B)|}}\int_{\mathbb R^n}\left|f(x)\overline{\phi(x-u)}\right|dx\right)^2\\\
&&=\frac{1}{(2\pi)^{n}{|det(B)|}}\left(\int_{\mathbb R^n}\left|f(x)\overline{\phi(x-u)}\right|dx\right)^2\\\
&&=\frac{1}{(2\pi)^{n}{|det(B)|}}\left(\int_{\mathbb R^n}\left|f(x)\right|^2dx\right)\left(\int_{\mathbb R^n}\left|\overline{\phi(x-u)}\right|^2dx\right)\\\
&&=\frac{1}{(2\pi)^{n}{|det(B)|}}\|f\|^2_{L^2(\mathbb R^n)}\|\phi\|^2_{L^2(\mathbb R^n)},
\end{eqnarray*}
On further simplification, we obtain
\begin{equation*}
\left|\mathcal V^{\bf M}_\phi[f](w,u)\right|\le\frac{1}{(2\pi)^{n/2}\sqrt{|det(B)|}}\|f\|_{L^2(\mathbb R^n)}\|\phi\|_{L^2(\mathbb R^n)},
\end{equation*}
which completes the proof.\\
\end{proof}
In the following theorem, we show that the proposed ST-NSLCT  is reversible in the sense
that the input signal $f$ can be recovered easily from the transformed domain.
\begin{theorem}[Inversion] For any fixed window function $\phi\in L^2(\mathbb R^n).$ Then, any $f\in L^2(\mathbb R^n)$ can be reconstructed by the formula
\begin{equation}
f(x)=\frac{1}{\|\phi(x)\|^2}\int_{\mathbb R^n}\int_{\mathbb R^n}\mathcal V^{\bf M}_{\phi}[f](w,u)\phi(x-u)\mathcal K_{\bf M^{-1}}(x,w)dwdu
\end{equation}

\end{theorem}
\begin{proof}
From (\ref{f-phi}), we have

\begin{eqnarray*}
f(x)\overline{\phi(x-u)}
&=&\mathcal L_{\bf M^{-1}}\left\{\mathcal V^{\bf M}_{\phi}[f](w,u)\right\}\\\
&=&f(x)\overline{\phi(x-u)}\\\
&=&\int_{\mathbb R^n}\mathcal V^{\bf M}_{\phi}[f](w,u)\mathcal K_{\bf M^{-1}}(x,w)dw\\\
\end{eqnarray*}
Multiplying both sides of last equation  by $\phi(x-u)$ and integrating
with respect to $du$, yields
\begin{eqnarray*}
\int_{\mathbb R^n}f(x)\overline{\phi(x-u)}\phi(x-u)du&=&\int_{\mathbb R^n}\int_{\mathbb R^n}\mathcal V^{\bf M}_{\phi}[f](w,u)\mathcal K_{\bf M^{-1}}(x,w)\phi(x-u)dwdu\\\
\int_{\mathbb R^n}f(x)|\phi(x-u)|^2du&=&\int_{\mathbb R^n}\int_{\mathbb R^n}\mathcal V^{\bf M}_{\phi}[f](w,u)\phi(x-u)\mathcal K_{\bf M^{-1}}(x,w)dwdu\\\
f(x)\|\phi(x)\|^2&=&\int_{\mathbb R^n}\int_{\mathbb R^n}\mathcal V^{\bf M}_{\phi}[f](w,u)\phi(x-u)\mathcal K_{\bf M^{-1}}(x,w)dwdu,\\\
\end{eqnarray*}
which implies
\begin{equation*}
f(x)=\frac{1}{\|\phi(x)\|^2}\int_{\mathbb R^n}\int_{\mathbb R^n}\mathcal V^{\bf M}_{\phi}[f](w,u)\phi(x-u)\mathcal K_{\bf M^{-1}}(x,w)dwdu.\\
\end{equation*}
\end{proof}

\begin{theorem}[Moyal's formula] Let $\mathcal V^{\bf M}_{\phi_1}[f](w,u$ and $\mathcal V^{\bf M}_{\phi_2}[g](w,u)$ be the ST-NSLCT transforms of $f$ and $g$, respectively. Then, we have
\begin{equation}
\left\langle \mathcal V^{\bf M}_{\phi_1}[f](w,u),\mathcal V^{\bf M}_{\phi_2}[g](w,u)\right\rangle=\langle f,g\rangle_{L^2(\mathbb R^n)}\langle \phi_1,\phi_2\rangle_{L^2(\mathbb R^n)}
\end{equation}
\end{theorem}
\begin{proof}
\begin{eqnarray*}
&&\left\langle \mathcal V^{\bf M}_{\phi_1}[f](w,u),\mathcal V^{\bf M}_{\phi_2}[g](w,u)\right\rangle\\\\
&&=\int_{\mathbb R^n}\int_{\mathbb R^n}\mathcal V^{\bf M}_{\phi_1}[f](w,u)\mathcal V^{\bf M}_{\phi_2}[g](w,u)dwdu\\\\
&&=\int_{\mathbb R^n}\int_{\mathbb R^n}\mathcal V^{\bf M}_{\phi_1}[f](w,u)\left[\overline{\int_{\mathbb R^n}g(x)\overline{\phi_2(x-u)}\mathcal K_{\bf M}(x,w)dx}\right]dwdu\\\\
&&=\int_{\mathbb R^n}\int_{\mathbb R^n}\int_{\mathbb R^n}\mathcal V^{\bf M}_{\phi_1}[f](w,u)\overline{\mathcal K_{\bf M}(x,w)}.\overline{g(x)}\phi_2(x-u)dxdudw\\\
&&=\int_{\mathbb R^n}\int_{\mathbb R^n}\left[\int_{\mathbb R^n}\mathcal V^{\bf M}_{\phi_1}[f](w,u)\overline{\mathcal K_{\bf M}(x,w)}dw\right]\overline{g(x)}\phi_2(x-u)dxdu.\\\
\end{eqnarray*}
Now applying (\ref{f-phi}) in above equation, we obtain
\begin{eqnarray*}
\left\langle \mathcal V^{\bf M}_{\phi_1}[f](w,u),\mathcal V^{\bf M}_{\phi_2}[g](w,u)\right\rangle
&=&\int_{\mathbb R^n}\int_{\mathbb R^n}f(x)\overline{\phi_1(x-u)}.\overline{g(x)}\phi_2(x-u)dxdu\\\\
&=&\int_{\mathbb R^n}f(x)\overline{g(x)}dx\int_{\mathbb R^n}\overline{\phi_1(x-u)}\phi_2(x-u)du\\\\
&=&\langle f,g\rangle_{L^2(\mathbb R^n)}\langle \phi_1,\phi_2\rangle_{L^2(\mathbb R^n)},
\end{eqnarray*}
which completes the proof.
\end{proof}
From the above theorem, we obtain the following consequences.

(i) If $\phi_1=\phi_2,$ then
\begin{equation}\label{con1}
\left\langle \mathcal V^{\bf M}_{\phi}[f](w,u),\mathcal V^{\bf M}_{\phi}[g](w,u)\right\rangle=\| \phi\|_{L^2(\mathbb R^n)}\langle f,g\rangle_{L^2(\mathbb R^n)}.\\\\
\end{equation}

(ii) If $f=g$ and $\phi_1=\phi_2,$ then
\begin{equation}\label{con2}
\left\langle \mathcal V^{\bf M}_{\phi}[f](w,u),\mathcal V^{\bf M}_{\phi}[f](w,u)\right\rangle=\| \phi\|_{L^2(\mathbb R^n)}\| f\|_{L^2(\mathbb R^n)}.
\end{equation}

(iii) If $f=g$ and $\phi_1=\phi_2=1,$ then
\begin{equation}\label{con3}
\left\langle \mathcal V^{\bf M}_{\phi}[f](w,u),\mathcal V^{\bf M}_{\phi}[f](w,u)\right\rangle=\| f\|_{L^2(\mathbb R^n)}.
\end{equation}
Equation (\ref{con3}) states that the proposed ST-NSLCT(\ref{eqn M-WLCT}) becomes an isometry from $L^2(\mathbb R^n)$ into $L^2(\mathbb R^n)$. In
other words, the total energy of a  signal computed in the in the
short-time non-seperable linear canonical domain is equal to the total energy computed in
the spatial domain.

\section{ Uncertainty Principles for the Short-time Non-seperable  Linear canonical Transformation}\label{sec 4}
The uncertainty principle lies at the heart of harmonic analysis, which asserts that
“the position and the velocity of a particle cannot be both determined precisely at the
same time”. Authors in \cite{NDself,ND12} studied various uncertainty principles associated with non-seperable LCT, since ST-NSLCT is the generalized version of NSLCT and keeping in view the fact that the theory of UP's for the short-time non-separable
linear canonical transform is yet to be explored exclusively; therefore, it is natural and interesting to study different forms of UP's in the ST-NSLCT domain.

\subsection{ Pitt’s inequality}
The Pitt’s inequality in the Fourier domain expresses a fundamental relationship between
a sufficiently smooth function and the corresponding Fourier transform (Beckner 1995)\cite{ND23}. In \cite{NDself} authors introduced the Pitt’s inequality can be extended to the free metaplectic transformation domain, we here extend it to the ST-NSLCT domain as:
\begin{lemma}[Pitt’s inequality for the NSLCT \cite{MDL}]\label{lem pit} Let $f\in \mathbb S(\mathbb R^n)$ the Schwartz class in $L^2(\mathbb R^n)$, then we have the inequality
\begin{equation*}
\int_{\mathbb R^n}|w|^{-\alpha}|\mathcal L_{\bf M}[f](w)|^2dw\le C_{\alpha}|det(B)|^{-\alpha}\int_{\mathbb R^n}|x|^{\alpha}|f(x)|^2dx,\\\\
\end{equation*}
where $C_{\alpha}=\pi^{\alpha}\left[\Gamma\left(\frac{n-\alpha}{4}\right)/\Gamma\left(\frac{n+\alpha}{4}\right)\right]^2$ and $0\le\alpha < n.$
\end{lemma}

Based on Pitt’s inequality for the NSLCT, we obtain Pitt’s
inequality for the ST-NSLCT.
\begin{theorem}[Pitt’s inequality for the ST-NSLCT] Under the assumptions of Lemma \ref{lem pit}, we have
\begin{eqnarray}
&&\int_{\mathbb R^n}\int_{\mathbb R^n}|w|^{-\alpha}|\mathcal V^{\bf M}_{\phi}[f](w,u)|^2dwdu\\\
&&=C_{\alpha}|det(B)|^{-\alpha}\|\phi\|^2_{L^2(\mathbb R^n)}\int_{\mathbb R^n}|x|^{\alpha}|f(x)|^2dx.
\end{eqnarray}
\end{theorem}

\begin{proof} Let $f^\phi_u(x)=f(x)\overline{\phi(x-u)}$
then by virtue of  (\ref{eqn WLCT-LCT}), we have
\begin{eqnarray*}
\int_{\mathbb R^n}|w|^{-\alpha}|\mathcal V^{\bf M}_{\phi}[f](w,u)|^2dw=\int_{\mathbb R^n}|w|^{-\alpha}|\mathcal L_{\bf M}[f^\phi_u(x)]|^2dw.\\\
\end{eqnarray*}
Now applying Lemma \ref{lem pit} to the L.H.S of above equation, we obtain
\begin{eqnarray}\label{p1}
&&\nonumber\int_{\mathbb R^n}|w|^{-\alpha}|\mathcal V^{\bf M}_{\phi}[f](w,u)|^2dw\\\
\nonumber&&\le C_{\alpha}|det(B)|^{-\alpha}\int_{\mathbb R^n}|x|^{\alpha}|f^\phi_u(x)|^2dx\\\
\nonumber&&\le C_{\alpha}|det(B)|^{-\alpha}\int_{\mathbb R^n}|x|^{\alpha}|f(x)\overline{\phi(x-u)}|^2dx.\\\
\end{eqnarray}
On integrating both sides of (\ref{p1})with respect to $du$, we have
\begin{eqnarray*}
&&\int_{\mathbb R^n}\int_{\mathbb R^n}|w|^{-\alpha}|\mathcal V^{\bf M}_{\phi}[f](w,u)|^2dwdu\\\
&&\le C_{\alpha}|det(B)|^{-\alpha}\int_{\mathbb R^n}\int_{\mathbb R^n}|x|^{\alpha}|f(x)\overline{\phi(x-u)}|^2dxdu\\\
&&=C_{\alpha}|det(B)|^{-\alpha}\|\phi\|^2_{L^2(\mathbb R^n)}\int_{\mathbb R^n}|x|^{\alpha}|f(x)|^2dx,\\\
\end{eqnarray*}
which completes the proof.
\end{proof}

\subsection{Lieb’s UP} Here we shall establish  Lieb’s
Uncertainty Principle for the ST-NSLCT by using the relation between ST-NSLCT with the short-time Fourier transform in $L^2(\mathbb R^n).$

 \begin{theorem}[Lieb's] Let $\phi,f\in L^2(\mathbb R^n)$ and $2\le p<\infty$. Then following inequality holds
 \begin{eqnarray*}
 \int_{\mathbb R^n}\int_{\mathbb R^n}|\mathcal V^{\bf M}_{\phi}[f](w,u)|^pdwdx
 \le\frac{2}{p}\frac{|B|}{{|det(A)|}^{p/2}}(\|f\|_{L^2(\mathbb R^n)}\|\phi\|_{L^2(\mathbb R^n)})^p.\\
 \end{eqnarray*}
 \end{theorem}
 \begin{proof}
 For every $f\in \mathbb S(\mathbb R^n)\subseteq L^2(\mathbb R^n),$ the Leib's inequality for the STFT states that
 \begin{equation}\label{l1}
 \int_{\mathbb R^n}\int_{\mathbb R^n}|\mathcal V_{\phi}[f](w,u)|^pdwdx\le\frac{2}{p}\left(\|f\|_{L^2(\mathbb R^n)}\|\phi\|_{L^2(\mathbb R^n)}\right)^p,
 \end{equation}
 where $\mathcal V_{\phi}[f](w,u)$ denotes the STFT of $f$ given by (\ref{def STFT}) and $\mathbb S(\mathbb R^n)$ is the Schwartz class in $L^2(\mathbb R^n)$. To obtain an analogue of the Leib's inequality for the ST-NSLCT, we replace $f$ in above equation by $H$ [defined in \ref{fun H}], we have
 \begin{eqnarray}
 \nonumber&&\int_{\mathbb R^n}\int_{\mathbb R^n}|\mathcal V_{\phi}[H](w,u)|^pdwdx\\\
 \nonumber&&\le\frac{2}{p}\left(\|H\|_{L^2(\mathbb R^n)}\|\phi\|_{L^2(\mathbb R^n)}\right)^p\\\
 \label{l2}&&=\frac{2}{p}\left(\left(\int_{\mathbb R^n}\left|\frac{1}{\sqrt{|det(A)|}}e^{\frac{i(x^TB^{-1}Ax)}{2}}f(x)\right|^2dx\right)^{1/2}\|\phi\|_{L^2(\mathbb R^n)}\right)^p
\end{eqnarray}
 Setting $w=B^{-1}w$ in (\ref{l2}), we get

 \begin{eqnarray}
 \nonumber &&\int_{\mathbb R^n}\int_{\mathbb R^n}|B^{-1}||\mathcal V_{\phi}[H](B^{-1}w,u)|^pdwdx\\\
\label{l3} &&\le\frac{2}{p}\frac{1}{{|det(A)|}^{p/2}}\left(\left(\int_{\mathbb R^n}\left|f(x)\right|^2dx\right)^{1/2}\|\phi\|_{L^2(\mathbb R^n)}\right)^p.
 \end{eqnarray}

 Now using Lemma \ref{lem WLCT-WFT} to L.H.S of (\ref{l3}), we obtain
 \begin{eqnarray*}
 &&\int_{\mathbb R^n}\int_{\mathbb R^n}|B^{-1}||e^{\frac{-iw^TDB^{-1}w}{2}}\mathcal V^{\bf M}_{\phi}[f](w,u)|^pdwdx\\\
 &&\le\frac{2}{p}\frac{1}{{|det(A)|}^{p/2}}\left(\left(\int_{\mathbb R^n}\left|f(x)\right|^2dx\right)^{1/2}\|\phi\|_{L^2(\mathbb R^n)}\right)^p\\\,
 \end{eqnarray*}

 on further simplifying, we have
 \begin{eqnarray*}
 &&\int_{\mathbb R^n}\int_{\mathbb R^n}|\mathcal V^{\bf M}_{\phi}[f](w,u)|^pdwdx\\\
 &&\le\frac{2|B|}{p}\frac{1}{{|det(A)|}^{p/2}}\left(\left(\int_{\mathbb R^n}\left|f(x)\right|^2dx\right)^{1/2}\|\phi\|_{L^2(\mathbb R^n)}\right)^p\\\
 &&=\frac{2}{p}\frac{|B|}{{|det(A)|}^{p/2}}(\|f\|_{L^2(\mathbb R^n)}\|\phi\|_{L^2(\mathbb R^n)})^p.
 \end{eqnarray*}
 Which completes the proof.

 \end{proof}

\subsection{Heisenberg’s UP}
 In \cite{NDself}, authors introduced Heisenberg’s UP to observe
the lower bound of uncertainty corresponding to the free metaplectic transformation. We extend it to the propose Heisenberg’s UP fot the  ST-NSLCT.

\begin{theorem}[Heisenberg’s UP] Let $\mathcal V^{\bf M}_{\phi}[f]$ be the short time  non-separable linear canonical transform of any
non-trivial function $f\in L^2(\mathbb R^n)$ with respect to a real free symplectic matrix ${\bf M}=(A,B:C,D),$ then following inequality holds:
\begin{eqnarray}
\nonumber&&\left\{\int_{\mathbb R^n}\int_{\mathbb R^n}|w|^2\left|\mathcal V^{\bf M}_{\phi}[f](w,u)\right|^2dw\right\}^{1/2}\left\{\int_{\mathbb R^n}x^2\left|f(x)\right|^2dx\right\}^{1/2}\\
\label{eqn hsb}&&\ge\frac{n\sigma_{min}(B)}{4\pi}\|f\|^2_{L^2(\mathbb R^n)}\|\phi\|_{L^2(\mathbb R^n)}.
\end{eqnarray}
\end{theorem}
\begin{proof}
 For any $f\in L^2(\mathbb R^n)$ the  Heisenberg–Pauli–Weyl uncertainty inequality for NSLCT domain is given by \cite{NDself}:
\begin{equation*}
\int_{\mathbb R^n}x^2|f(x)|^2dx\int_{\mathbb R^n}|w|^2|\mathcal L_{\bf M}[f](w)|^2dw\ge\frac{n^2\sigma^2_{min}(B)}{16\pi^2}\left\{\int_{\mathbb R^n}|f(x)|^2dx\right\}^2
\end{equation*}
Using the inversion of NSLCT (\ref{inv MLCT}) into LHS and (\ref{eqn par2}) to the RHS, above inequality becomes
\begin{eqnarray*}
&&\int_{\mathbb R^n}x^2\left|\mathcal L_{\bf M^{-1}}\{\mathcal L_{\bf M}[f](w)\}\right|^2dx\int_{\mathbb R^n}|w|^2\left|\mathcal L_{\bf M}[f](w)\right|^2dw\\\
&&\ge\frac{n^2\sigma^2_{min}(B)}{16\pi^2}\left\{\int_{\mathbb R^n}\left|\mathcal L_{\bf M}[f](w)\right|^2dw\right\}^2\\\
\end{eqnarray*}
since $\mathcal V^{\bf M}_{\phi}[f]\in L^2(\mathbb R^n),$  therefore replacing $\mathcal L_{\bf M}[f](w)$ by $\mathcal V^{\bf M}_{\phi}[f]\in L^2(\mathbb R^n)$, we obtain
\begin{eqnarray*}
&&\int_{\mathbb R^n}x^2\left|\mathcal L_{\bf M^{-1}}\{\mathcal V^{\bf M}_{\phi}[f](w,u)\}\right|^2dx\int_{\mathbb R^n}|w|^2\left|\mathcal V^{\bf M}_{\phi}[f](w,u)\right|^2dw\\\
&&\ge\frac{n^2\sigma^2_{min}(B)}{16\pi^2}\left\{\int_{\mathbb R^n}\left|\mathcal V^{\bf M}_{\phi}[f](w,u)\right|^2dw\right\}^2\\\
\end{eqnarray*}
First taking square of above equation and then integrating both sides with respect $du$, we obtain
\begin{eqnarray*}
&&\int_{\mathbb R^n}\left\{\int_{\mathbb R^n}x^2\left|\mathcal L_{\bf M^{-1}}\{\mathcal V^{\bf M}_{\phi}[f](w,u)\}\right|^2dx\right\}^{1/2}\left\{\int_{\mathbb R^n}|w|^2\left|\mathcal V^{\bf M}_{\phi}[f](w,u)\right|^2dw\right\}^{1/2}du\\\
&&\ge\frac{n\sigma_{min}(B)}{4\pi}\int_{\mathbb R^n}\int_{\mathbb R^n}\left|\mathcal V^{\bf M}_{\phi}[f](w,u)\right|^2dwdu\\\
\end{eqnarray*}
As a consequence of the Cauchy–Schwartz’s inequality and Fubini theorem, above inequality yields
\begin{eqnarray*}
&&\left\{\int_{\mathbb R^n}\int_{\mathbb R^n}x^2\left|\mathcal L_{\bf M^{-1}}\{\mathcal V^{\bf M}_{\phi}[f](w,u)\}\right|^2dxdu\right\}^{1/2}\left\{\int_{\mathbb R^n}\int_{\mathbb R^n}|w|^2\left|\mathcal V^{\bf M}_{\phi}[f](w,u)\right|^2dw\right\}^{1/2}\\\
&&\ge\frac{n\sigma_{min}(B)}{4\pi}\int_{\mathbb R^n}\int_{\mathbb R^n}\left|\mathcal V^{\bf M}_{\phi}[f](w,u)\right|^2dwdu\\\
\end{eqnarray*}
Now using (\ref{f-phi}) in L.H.S and (\ref{con2}) in R.H.S of above equation, we have
\begin{eqnarray*}
&&\left\{\int_{\mathbb R^n}\int_{\mathbb R^n}x^2\left|f(x)\overline{\phi(x-u)}\right|^2dxdu\right\}^{1/2}\left\{\int_{\mathbb R^n}\int_{\mathbb R^n}|w|^2\left|\mathcal V^{\bf M}_{\phi}[f](w,u)\right|^2dw\right\}^{1/2}\\\
&&\ge\frac{n\sigma_{min}(B)}{4\pi}\|\phi\|^2_{L^2(\mathbb R^n)}\|f\|^2_{L^2(\mathbb R^n)}.\\\
\end{eqnarray*}
 Simplifying above, we obtain
\begin{eqnarray*}
&&\left\{\|\phi\|^2_{L^2(\mathbb R^n)}\int_{\mathbb R^n}x^2\left|f(x)\right|^2dx\right\}^{1/2}\left\{\int_{\mathbb R^n}\int_{\mathbb R^n}|w|^2\left|\mathcal V^{\bf M}_{\phi}[f](w,u)\right|^2dw\right\}^{1/2}\\\
&&\ge\frac{n\sigma_{min}(B)}{4\pi}\|\phi\|^2_{L^2(\mathbb R^n)}\|f\|^2_{L^2(\mathbb R^n)},\\\
\end{eqnarray*}
dividing both sides by $\|\phi\|_{L^2(\mathbb R^n)},$ we will get the desired result.
\end{proof}

\begin{theorem}[Hausdorff–Young] For
 Let $1\le p\le 2$ and $q$ be such that $\frac{1}{p}+\frac{1}{q}=1.$
Then for $\phi\in L^q(\mathbb R^n)$ and $f\in L^p(\mathbb R^n),$ following inequality holds\\\
\begin{equation}
\|\mathcal V^{\bf M}_{\phi}[f](w,u)\|_{L^q(\mathbb R^n)}\le\|\phi\|_{L^q(\mathbb R^n)}\|f\|_{L^p(\mathbb R^n)}.
\end{equation}
\end{theorem}
\begin{proof}
 The definition of NSLCT \ref{def mul lct} and procedure defined in \cite{bah2}[Theorem 5.1] together led to
\begin{equation}\label{h1}
\|\mathcal L_{\bf M}[f]\|_{L^q(\mathbb R^n)}\le \|f\|_{L^p(\mathbb R^n)}.
\end{equation}
For  $p=1$, we get
\begin{equation}\label{h2}
\|\mathcal L_{\bf M}[f]\|_{L^\infty(\mathbb R^n)}\le \|f\|_{L^1(\mathbb R^n)}.
\end{equation}
Now using (\ref{h2}) and taking $\|\phi\|_{L^q(\mathbb R^n)}=1$, equation (\ref{eqn WLCT-LCT}) yields
\begin{eqnarray*}
\|\mathcal V^{\bf M}_{\phi}[f](w,u)\|_{L^\infty(\mathbb R^n)}&=&\|\mathcal L_{\bf M}\{f(x)\overline{\phi(x-u)}\}\|_{L^\infty(\mathbb R^n)}\\
&\le&\|\{f(x)\overline{\phi(x-u)}\}\|_{L^1(\mathbb R^n)}\\
&\le&\|f\|_{L^1(\mathbb R^n)}\|{\phi}\|_{L^\infty(\mathbb R^n)}\\
&=&\|f\|_{L^1(\mathbb R^n)}.
\end{eqnarray*}
Now for $p=2$, we obtain
\begin{equation*}
\|\mathcal V^{\bf M}_{\phi}[f](w,u)\|_{L^2(\mathbb R^n)}\le\|f\|_{L^2(\mathbb R^n)}.
\end{equation*}
By Riesz–Thorin interpolation theorem, above yields

\begin{equation}\label{h36}
\|\mathcal V^{\bf M}_{\phi}[f](w,u)\|_{L^q(\mathbb R^n)}\le\|f\|_{L^p(\mathbb R^n)}.
\end{equation}
Now setting \quad$g=\frac{\phi}{\|\phi\|_{L^q(\mathbb R^n)}},$ where $\phi$ is a window function in $L^p(\mathbb R^n),$ we have by anti-linearity property of NSLCT
\begin{equation}\label{h37}
 \mathcal V^{\bf M}_{g}[f](w,u)=\frac{1}{\|\phi\|_{L^q(\mathbb R^n}}\mathcal V^{\bf M}_{\phi}[f](w,u),
 \end{equation}
 on taking $\phi=g$ in (\ref{h36}), we obtain
 \begin{equation}\label{h38}
\|\mathcal V^{\bf M}_{g}[f](w,u)\|_{L^q(\mathbb R^n)}\le\|f\|_{L^p(\mathbb R^n)}
\end{equation}
 which on simplification becomes
 \begin{equation*}
\|\mathcal V^{\bf M}_{\phi}[f](w,u)\|_{L^q(\mathbb R^n)}\le\|\phi\|_{L^q(\mathbb R^n)}\|f\|_{L^p(\mathbb R^n)}.
\end{equation*}
Hence completes the proof.
 \end{proof}

 \subsection{Logarithmic UP}

In 1995
\cite{WO36} W. Beckner  introduced Logarithmic uncertainty principle. In this subsection we obtain the concept of Logarithmic
uncertainty principle for the Short-time non-separable linear canonical transform as follows:\\

\begin{theorem}[Logarithmic uncertainty principle ]Let $\phi$ be a window function in $L^2(\mathbb R^n)$ and let $\mathcal V^{\bf M}_{\phi}[f](w,u)\in  \mathbb S(\mathbb R^n), $ then the short-time
non-separable linear canonical transform satisfies the
following logarithmic uncertainty inequality:
 \begin{eqnarray*}
\nonumber&&\int_{\mathbb R^n}\int_{\mathbb R^n}\ln|wB^{-T}||\mathcal V^{\bf M}_{\phi}[f](w,u)|^2dwdu+\|\phi\|^2_{L^2(\mathbb R^n)}\int_{\mathbb R^n}\ln|x||f(x)|^2dx\\\
\label{eqn log}&&\ge\left(\frac{\Gamma'(n/2)}{\Gamma(n/2)}-\ln\pi\right)\|\phi\|^2_{L^2(\mathbb R^n)}\|f\|^2_{L^2(\mathbb R^n)}.
\end{eqnarray*}
\end{theorem}
\begin{proof}
For any $f\in \mathbb S(\mathbb R^n)$ the  logarithmic uncertainty principle for the non-separable linear canonical transform domain is given by \cite{NDself}:
\begin{eqnarray*}
\nonumber&&\int_{\mathbb R^n}\ln|x||f(x)|^2dx+\int_{\mathbb R^n}\ln|wB^{-T}||\mathcal L_{\bf M}[f](w)|^2dw\\\
\label{log}&&\ge\left(\frac{\Gamma'(n/2)}{\Gamma(n/2)}-\ln\pi\right)\int_{\mathbb R^n}|f(x)|^2dx.
\end{eqnarray*}
Now invoking the inversion formula of the non-separable linear canonical transform on the L.H.S and Parseval's formula  on R.H.S, we obtain
\begin{eqnarray*}
\nonumber&&\int_{\mathbb R^n}\ln|x||\mathcal L_{\bf M^{-1}}\{\mathcal L_{\bf M}[f](w)\}|^2dx+\int_{\mathbb R^n}\ln|wB^{-T}||\mathcal L_{\bf M}[f](w)|^2dw\\\
\nonumber&&\ge\left(\frac{\Gamma'(n/2)}{\Gamma(n/2)}-\ln\pi\right)\int_{\mathbb R^n}|\mathcal L_{\bf M}[f](w)|^2dw.
\end{eqnarray*}
Since both $\mathcal L_{\bf M}[f](w)$ and $\mathcal V^{\bf M}_{\phi}[f](w,u)$ are in $\mathbb S(\mathbb R^n)$ thus we can replace $\mathcal L_{\bf M}[f](w)$ by $\mathcal V^{\bf M}_{\phi}[f](w,u)$
on the both sides of above, to get
\begin{eqnarray*}
\nonumber&&\int_{\mathbb R^n}\ln|x||\mathcal L_{\bf M^{-1}}\{\mathcal V^{\bf M}_{\phi}[f](w,u)\}|^2dx+\int_{\mathbb R^n}\ln|wB^{-T}||\mathcal V^{\bf M}_{\phi}[f](w,u)|^2dw\\\
\nonumber&&\ge\left(\frac{\Gamma'(n/2)}{\Gamma(n/2)}-\ln\pi\right)\int_{\mathbb R^n}|\mathcal V^{\bf M}_{\phi}[f](w,u)|^2dw.
\end{eqnarray*}
Integrating above inequality with respect $du$ on both sides and then by virtue of Fubini's theorem, we have

\begin{eqnarray*}
\nonumber&&\int_{\mathbb R^n}\int_{\mathbb R^n}\ln|x||\mathcal L_{\bf M^{-1}}\{\mathcal V^{\bf M}_{\phi}[f](w,u)\}|^2dxdu+\int_{\mathbb R^n}\int_{\mathbb R^n}\ln|wB^{-T}||\mathcal V^{\bf M}_{\phi}[f](w,u)|^2dwdu\\\
\nonumber&&\ge\left(\frac{\Gamma'(n/2)}{\Gamma(n/2)}-\ln\pi\right)\int_{\mathbb R^n}\int_{\mathbb R^n}|\mathcal V^{\bf M}_{\phi}[f](w,u)|^2dwdu.
\end{eqnarray*}
Now using (\ref{f-phi}) in L.H.S and (\ref{con2}) in R.H.S of above inequality, we have
\begin{eqnarray*}
\nonumber&&\int_{\mathbb R^n}\int_{\mathbb R^n}\ln|wB^{-T}||\mathcal V^{\bf M}_{\phi}[f](w,u)|^2dwdu+\int_{\mathbb R^n}\int_{\mathbb R^n}\ln|x||f(x)\phi(x-u)|^2dxdu\\\
\nonumber&&\ge\left(\frac{\Gamma'(n/2)}{\Gamma(n/2)}-\ln\pi\right)\|\phi\|^2_{L^2(\mathbb R^n)}\|f\|^2_{L^2(\mathbb R^n)}.
\end{eqnarray*}
Further simplifying left-hand side, we obtain the desired result as
\begin{eqnarray*}
\nonumber&&\int_{\mathbb R^n}\int_{\mathbb R^n}\ln|wB^{-T}||\mathcal V^{\bf M}_{\phi}[f](w,u)|^2dwdu+\|\phi\|^2_{L^2(\mathbb R^n)}\int_{\mathbb R^n}\ln|x||f(x)|^2dx\\\
\nonumber&&\ge\left(\frac{\Gamma'(n/2)}{\Gamma(n/2)}-\ln\pi\right)\|\phi\|^2_{L^2(\mathbb R^n)}\|f\|^2_{L^2(\mathbb R^n)}.
\end{eqnarray*}

\end{proof}
\subsection{Hardy’s UP}
G.H. Hardy first introduced the Hardy’s uncertainty principle in 1933 \cite{ND19}. Hardy’s uncertainty principle says that it
is impossible for a function and its Fourier transform to decrease very rapidly simultaneously. Hardy’s UP in the Fourier transform domain \cite{ND20} was
given as follows.

\begin{lemma}[Hardy’s UP in the Fourier transform \cite{ND20}]\label{hardy FT}
If $f(x)=\mathcal O(e^{-|x|^2/\beta^2})$, $\mathcal F[f](w)=\mathcal O((2\pi)^{n/2}e^{-16\pi^2|w|^2/\alpha^2})$  and $1/\alpha\beta>1/4, $ then $f\equiv 0.$ If $1/\alpha\beta=1/4,$ then
\begin{equation*}
f=Ce^{-|x|^2/\beta^2}.
\end{equation*}
Where $C$ is a constant in $\mathbb C$
\end{lemma}

Based on Lemma \ref{hardy FT}, we derive the corresponding Hardy’s UP for the ST-NSLCT.

\begin{theorem}[Hardy’s UP in the ST-NSLCT]
If $f(x)=\mathcal O(e^{-|x|^2/\beta^2})$, $\mathcal F[f](w)=\mathcal O((2\pi)^{n/2}e^{-16\pi^2|B^{-1}w|^2/\alpha^2})$  and $1/\alpha\beta>1/4, $ then $f\equiv 0.$ If $1/\alpha\beta=1/4,$ then
\begin{equation*}
f=Ce^{-|x|^2/\beta^2-\frac{i(x^TB^{-1}Ax)}{2}}.
\end{equation*}
Where $C$ is a constant in $\mathbb C$
\end{theorem}
\begin{proof}
Consider then function $F(x)=e^{\frac{i(x^TB^{-1}Ax)}{2}}f(x)\overline{\phi(x-u)}$.\\ On setting $u=x$ then  we have
\begin{equation*}
|F(x)|=|f(x)||\overline{\phi(0)}|=\mathcal O(e^{-|x|^2/\beta^2})
\end{equation*}

Also by Lemma \ref{lem NWLCT-FT mod}, we have
\begin{eqnarray*}
|\mathcal F[F](w)|&=&\sqrt{det(B)}|\mathcal V^{\bf M}_{\phi}[f](Bw,u)|\\
&=&\mathcal O((2\pi)^{n/2}e^{-16\pi^2|B^{-1}w|^2/\alpha^2}).
\end{eqnarray*}
Following from Lemma \ref{hardy FT}:\\\
 If $1/\alpha\beta>1/4, $ we have $F\equiv 0$ which implies $f\equiv 0$.\\\
Also,if $1/\alpha\beta=1/4,$ then $$F(x)=Ce^{-|x|^2/\beta^2}$$ implies $$f(x)=Ce^{-|x|^2/\beta^2-\frac{i(x^TB^{-1}Ax)}{2}}.$$
\end{proof}
This completes the proof.

\subsection{Beurling’s UP}
Beurling’s uncertainty principle is a more general version of Hardy’s uncertainty principle, which is given by A.Beurling. It implies the weak form of Hardy’s UP immediately.
Beurling’s UP in the Fourier transform domain is as follows

\begin{lemma}[Beurling’s UP in the Fourier transform domain\cite{ND21}]\label{beur ucp FT} Let $f\in L^2(\mathbb R^n)$ and $d\ge 0$ satisfy
\begin{equation*}
\int_{\mathbb R^n}\int_{\mathbb R^n}\dfrac{|f(x)||\mathcal F[f](w)|}{(1+\|x\|+\|w\|)^d}e^{2\pi|\langle x,w\rangle|}dxdw\le\infty.
\end{equation*}
then $$ f(x)=P(x)e^{-\pi\langle Ax,x\rangle},$$
where $A$ is a real positive definite symmetric matrix and $P(x)$ is a polynomial of degree $<\frac{d-n}{2}.$
\end{lemma}
According to Lemma \ref{beur ucp FT}, we derive Beurling’s uncertainty principle for the ST-NSLCT.
\begin{theorem}[Beurling’s UP ]\label{beur ucp} Let $f,\phi\in L^2(\mathbb R^n)$ where $\phi$ is a nonzero window function  and  $\mathcal V^{\bf M}_{\phi}[f]\in L^2(\mathbb R^n)$  satisfy
\begin{equation*}
\int_{\mathbb R^n}\int_{\mathbb R^n}\dfrac{|f(x)||\overline{\phi(x-u)}||\mathcal V^{\bf M}_{\phi}[f](w,u)|}{\sqrt{|det(B)|}(1+\|x\|+\|B^{-1}w\|)^d}e^{2\pi|\langle x,B^{-1}w\rangle|}dxdw <\infty,\quad where\quad d\ge 0.\\\
\end{equation*}
Then $$ f(x)=P(x)/\overline{\phi(x-u)}e^{\frac{-i(x^TB^{-1}Ax)}{2}-\pi\langle Ax,x\rangle},$$
where $A$ is a real positive definite symmetric matrix and $P(x)$ is a polynomial of degree $<\frac{d-n}{2}.$
\end{theorem}

\begin{proof}
Consider then function $F(x)=e^{\frac{i(x^TB^{-1}Ax)}{2}}f(x)\overline{\phi(x-u)}$, then

\begin{eqnarray*}
&&\int_{\mathbb R^n}\int_{\mathbb R^n}\dfrac{|F(x)||\mathcal F[F](w)|}{(1+\|x\|+\|w\|)^d}e^{2\pi|\langle x,w\rangle|}dxdw\\\
&&=\int_{\mathbb R^n}\int_{\mathbb R^n}\dfrac{|f(x)||\overline{\phi(x-u)}|\sqrt{|det(B)|}|\mathcal V^{\bf M}_{\phi}[f](Bw,u)|}{(1+\|x\|+\|w\|)^d}e^{2\pi|\langle x,w\rangle|}dxdw\\\
&&=\int_{\mathbb R^n}\int_{\mathbb R^n}\dfrac{|f(x)||\overline{\phi(x-u)}||\mathcal V^{\bf M}_{\phi}[f](w,u)|}{\sqrt{|det(B)|}(1+\|x\|+\|B^{-1}w\|)^d}e^{2\pi|\langle x,B^{-1}w\rangle|}dxdw <\infty.\\\
\end{eqnarray*}
Therefore by Lemma \ref{beur ucp FT}, we have $$ F(x)=P(x)e^{-\pi\langle Ax,x\rangle},$$ where $A$ is a real positive definite symmetric matrix and $P(x)$ is a polynomial of degree $<\frac{d-n}{2}.$
Furthermore $f(x)=P(x)/\overline{\phi(x-u)}e^{\frac{-i(x^TB^{-1}Ax)}{2}-\pi\langle Ax,x\rangle}.$
\end{proof}
This completes the proof.

\subsection{Nazarov’s UP for the ST-NSLCT}

 Nazarov’s UP was first proposed by F.L. Nazarov in 1993 \cite{ND17}. It
measures the localization of a nonzero function by taking into consideration the notion of support of the function instead of the dispersion. In other words it argues what happens if a non-trival function and its Fourier transform are only small outside a compact
set.Let us start with Nazarov’s UP for the Fourier transform.
\begin{lemma}[Nazarov’s UP for the Fourier transform\cite{ND18}]\label{lem naz ft}
There exists a constant $K$, such that for finite Lebesgue measurable sets $S$, $E\subset\mathbb R^n$ and for every $f\in L^2(\mathbb R^n),$ we have
\begin{equation*}
Ke^{K(S,E)}\left(\int_{\mathbb R^n\backslash S}|f(x)|^2dx+\int_{\mathbb R^n\backslash E}|\mathcal F[f](w)|^2dw\right)\ge\int_{\mathbb R^n}|f(x)|^2dx
\end{equation*}
where $K(S,E)=
{Kmin(|S||E|,|S|^{1/n},\mu(E),\mu(S)|E|^{1/n})},$ $\mu(S)$ is the mean width of $S$ and $|S|$ denotes the Lebesgue measure of $S.$
\end{lemma}

Now we shall establish Nazarov’s UP to the short-time non-sperable linear canonical transform domain.
\begin{theorem}[Nazarov’s UP] Let $\mathcal V^{\bf M}_{\phi}[f]$ be the short-stime NSLCT, then under the assumptions of Lemma \ref{lem naz ft} for every $f\in L^2(\mathbb R^n)$, the following inequality holds:
\begin{eqnarray}\label{naz}
\nonumber&&\|\phi\|^2_{L^2(\mathbb R^n)}\int_{\mathbb R^n}|f(x)|^2dx\\\
\nonumber&&\le  Ke^{K(S,E)}\left(\|\phi\|^2_{L^2(\mathbb R^n)}\int_{\mathbb R^n\backslash S}|f(x)|^2dx+\int_{\mathbb R^n}\int_{\mathbb R^n(\backslash E  B)}|\mathcal V^{\bf M}_{\phi}[f](w,u)|^2dwdu\right).
\end{eqnarray}
\end{theorem}
\begin{proof}
Applying Lemma \ref{lem naz ft} to the function $F(x)\in L^2(\mathbb R^n)$ defined in Lemma \ref{lem NWLCT-FT mod}, we have
\begin{equation*}
\int_{\mathbb R^n}|F(x)|^2dx\le Ke^{K(S,E)}\left(\int_{\mathbb R^n\backslash S}|F(x)|^2dx+\int_{\mathbb R^n\backslash E}|\mathcal F[F](w)|^2dw\right)
\end{equation*}
Now with the help of Lemma \ref{eqn NWLCT-FT mod}, above yields
\begin{eqnarray*}
&&\int_{\mathbb R^n}|f(x)\overline{\phi(x-u)}|^2dx\\\
&&\le Ke^{K(S,E)}\left(\int_{\mathbb R^n\backslash S}|f(x)\overline{\phi(x-u)}|^2dx+\int_{\mathbb R^n\backslash E}|\sqrt{|det(B)|}\mathcal V^{\bf M}_{\phi}(Bw,u)|^2dw\right).
\end{eqnarray*}
  Integrating above equation both sides with respect to $u$, we obtain
  \begin{eqnarray*}
&&\int_{\mathbb R^n}\int_{\mathbb R^n}|f(x)\overline{\phi(x-u)}|^2dxdu\\\
&&\le Ke^{K(S,E)}\int_{\mathbb R^n}\left(\int_{\mathbb R^n\backslash S}|f(x)\overline{\phi(x-u)}|^2dx+\int_{\mathbb R^n\backslash E}|\sqrt{|det(B)|}\mathcal V^{\bf M}_{\phi}(Bw,u)|^2dw\right)du.
\end{eqnarray*}
Implementing Fubini's theorem, we have
\begin{eqnarray*}
&&\|\phi\|^2_{L^2(\mathbb R^n)}\int_{\mathbb R^n}|f(x)|^2dx\\\
&&\le Ke^{K(S,E)}\left(\|\phi\|^2_{L^2(\mathbb R^n)}\int_{\mathbb R^n\backslash S}|f(x)|^2dx+{|det(B)|}\int_{\mathbb R^n}\int_{\mathbb R^n\backslash E}|\mathcal V^{\bf M}_{\phi}(Bw,u)|^2dwdu\right)\\
&&= Ke^{K(S,E)}\left(\|\phi\|^2_{L^2(\mathbb R^n)}\int_{\mathbb R^n\backslash S}|f(x)|^2dx+\int_{\mathbb R^n}\int_{\mathbb R^n(\backslash E  B)}|\mathcal V^{\bf M}_{\phi}(w,u)|^2dwdu\right),
\end{eqnarray*}
which completes the proof.\\\\
\end{proof}

{\bf Alternative proof:}
For any function $f\in L^2(\mathbb R^n)$ and a pair of finite measurable subsets $S$ and
 $E$ of $\mathbb R^n$, Nazarov’s uncertainty principle in the linear canonical domain reads \cite{NDself}
\begin{equation}\label{nn1}
Ke^{K(S,E)}\left(\int_{\mathbb R^n\backslash S}|f(x)|^2dx+\int_{\mathbb R^n\backslash EB}|\mathcal L_{\bf M}[f](w)|^2dw\right)\ge\int_{\mathbb R^n}|f(x)|^2dx
\end{equation}
where $K(S,E)=
{Kmin(|S||E|,|S|^{1/n},\mu(E),\mu(S)|E|^{1/n})},$ $\mu(.)$ is the mean width of measurable subset, and $|.|$ denotes the Lebesgue measure.
Moreover, the relationship
between the ST-NSLCT and the NSLCT is given by
\begin{equation}\label{nn2}
\mathcal V^{\bf M}_{\phi}[f](w,u)=\mathcal L_{\bf M}\{f(x)\overline{\phi(x-u)}\}(w).
\end{equation}
Since $f(x)\overline{\phi(x-u)}\in L^2(\mathbb R^n)$ then by virtue of (\ref{nn1}), we have
\begin{eqnarray*}
&&\int_{\mathbb R^n}|f(x)\overline{\phi(x-u)}|^2dx\\\
&&\le Ke^{K(S,E)}\left(\int_{\mathbb R^n\backslash S}|f(x)\overline{\phi(x-u)}|^2dx+\int_{\mathbb R^n\backslash EB}|\mathcal L_{\bf M}[f(x)\overline{\phi(x-u)}](w)|^2dw\right)
\end{eqnarray*}
Using (\ref{nn2}), above yields
\begin{eqnarray*}
&&\int_{\mathbb R^n}|f(x)\overline{\phi(x-u)}|^2dx\\\
&&\le Ke^{K(S,E)}\left(\int_{\mathbb R^n\backslash S}|f(x)\overline{\phi(x-u)}|^2dx+\int_{\mathbb R^n\backslash EB}|\mathcal V^{\bf M}_{\phi}[f](w,u)|^2dw\right).
\end{eqnarray*}
Integrating above equation both sides with respect to $u$, we obtain
\begin{eqnarray*}
&&\int_{\mathbb R^n}\int_{\mathbb R^n}|f(x)\overline{\phi(x-u)}|^2dxdu\\\
&&\le Ke^{K(S,E)}\int_{\mathbb R^n}\left(\int_{\mathbb R^n\backslash S}|f(x)\overline{\phi(x-u)}|^2dx+\int_{\mathbb R^n\backslash EB}|\mathcal V^{\bf M}_{\phi}[f](w,u)|^2dw\right)du.
\end{eqnarray*}
On implementing the well known Fubini theorem in above equation, we obtain the desired result as
\begin{eqnarray*}
&&\|\phi\|^2_{L^2(\mathbb R^n)}\int_{\mathbb R^n}|f(x)|^2dx\\\
&&\le Ke^{K(S,E)}\left(\|\phi\|^2_{L^2(\mathbb R^n)}\int_{\mathbb R^n\backslash S}|f(x)|^2dx+\int_{\mathbb R^n}\int_{\mathbb R^n\backslash EB}|\mathcal V^{\bf M}_{\phi}(w,u)|^2dwdu\right)\\
\end{eqnarray*}
Which completes the proof.

\section{Conclusion}
In this paper we  presented a novel concept of short-time non-separable linear canonical transform. Based on the properties of ST-NSLCT and NSLCT, the relationship between these two notations are presented. Important properties such as   boundedness, reconstruction formula, Moyals formula and  are derived. Finally, we extend some different uncertainty principles (UP) from quantum mechanics
including Lieb's, Pitt's UP, Heisenberg’s uncertainty principle, Hausdorff-Young,  Hardy’s
UP, Beurling’s UP, Logarithmic UP, and Nazarov’s UP which have already been well studied in the
ST-NSLC domain.

\end{document}